\documentclass[11pt,letterpaper]{article}

\usepackage[margin=1in]{geometry}
\usepackage{theorem,latexsym,graphicx,amssymb}
\usepackage{amsmath,enumerate}
\usepackage{bbm}
\usepackage{subfigure}
\usepackage{float}
\usepackage{epsfig}
\usepackage{xspace}
\usepackage{paralist}
\usepackage{enumerate}
\usepackage{cases}
\usepackage{caption}
\usepackage{multicol}
\usepackage{graphicx}
\usepackage{xcolor}
\usepackage{tikz}
\usepackage{ifthen}
\usepackage{algorithm}
\usepackage[noend]{algpseudocode}

\usepackage{nicefrac}
\usepackage{xfrac}

\newenvironment{proof}{{\bf Proof:  }}{\hfill\rule{2mm}{2mm}\vspace*{5pt}}

\newtheorem{theorem}{Theorem}[section]
\newtheorem{lemma}{Lemma}[section]
\newtheorem{claim}{Claim}[section]
\newtheorem{fact}{Fact}[section]

\newcommand{\vect}[1]{\ensuremath{\vec{#1}}}

\newcommand{\alg}{\mathrm{ALG}}

\newcommand{\dd}[1]{\mathrm{d}#1}

\newcommand{\eps}{\varepsilon}

\usepackage{tcolorbox}

\title{Optimal Competitive Ratio of Two-sided Online Bipartite Matching}

\author{
Zhihao Gavin Tang
\thanks{Key Laboratory of Interdisciplinary Research of Computation and Economics, Shanghai University of Finance and Economics (\texttt{tang.zhihao@mail.shufe.edu.cn})}
}

\date{}

\begin{document}

\maketitle

\begin{abstract}
We establish an optimal upper bound (negative result) of $\sim 0.526$ on the competitive ratio of the fractional version of online bipartite matching with two-sided vertex arrivals, matching the lower bound (positive result) achieved by Wang and Wong (ICALP 2015), and Tang and Zhang (EC 2024).
\end{abstract}

\section{Introduction}
\label{sec:intro}
Since the seminal work of Karp, Vazirani, and Vazirani~\cite{stoc/KarpVV90}, online matching has been one of the most important and productive research areas in online algorithms. Notably, it has found an unexpected yet impactful application in online advertising~\cite{jacm/MehtaSVV07}.

The original model is now referred to as online bipartite matching with \emph{one-sided} vertex arrivals, since it assumes that one side of the underlying bipartite graph is known upfront, while the other side arrives online. Later, Wang and Wong~\cite{icalp/WangW15} generalized the model to \emph{two-sided} vertex arrivals, and Huang et al.~\cite{jacm/HuangKTWZZ20} further introduced the fully online model. These two models allow all vertices of the graph to arrive online and thereby extend the applicability of online matching theory to domains such as ride-hailing and ride-sharing.

In this work, we study the two-sided online bipartite matching problem. All vertices of an underlying bipartite graph $G=(L\cup R,E)$ are initially unknown and arrive online in sequence. Upon the arrival of a vertex, the edges between it and previously arrived vertices are revealed. The algorithm may either match it to an unmatched neighbor or leave it unmatched for future consideration. The goal is to maximize the size of the resulting matching.
The performance of an algorithm is measured via competitive analysis: an algorithm is said to be $\Gamma$-competitive if, for every instance, the size of the matching it produces is at least a $\Gamma$ fraction of the offline maximum matching.

The one-sided model is a special case of the two-sided model in which all vertices in $R$ arrive after those in $L$. Moreover, the two-sided model generalizes naturally to non-bipartite graphs, where it is referred to as online matching with general vertex arrivals.

We focus on the fractional variant of the problem, in which edges may be matched fractionally, subject to the constraint that the total matched value incident to any vertex does not exceed one. Wang and Wong~\cite{icalp/WangW15} designed a dual-based algorithm with competitive ratio $0.526$ for the fractional problem. Later, Tang and Zhang~\cite{ec/TangZ24} proposed a primal-based algorithm attaining the same ratio. On the negative side, Wang and Wong proved an impossibility bound of $0.625$, separating the problem from the one-sided model; subsequent works~\cite{algorithmica/BuchbinderST19,ec/TangZ24} improved this bound to $0.592$ and $0.584$, respectively.

To summarize, prior to this work, the optimal competitive ratio was known to lie between $0.526$ and $0.584$.
We close this gap by establishing a tight upper bound of $0.526$ on the competitive ratio.
\begin{theorem}
\label{thm:main}
	No algorithm is $\Gamma^*+\Omega(1)$ competitive for the two-sided fractional online bipartite matching problem, where 
\[
\Gamma^* := \max_{k\ge 1} \frac{1}{\left( \frac{k+1}{2}\right)^{\frac{k+1}{2k}} \left( \frac{k-1}{2}\right)^{\frac{k-1}{2k}}+1} \approx 0.526~.
\]
\end{theorem}

For the integral version of the problem, less progress has been made. The only non-trivial algorithm is due to Gamlath et al.~\cite{focs/GamlathKMSW19}, achieving a competitive ratio of $0.5+\Omega(1)$. We remark that impossibility results against fractional algorithms automatically extend to randomized integral algorithms. Thus, our result shows that the room for improvement is small.

We compare the two-sided arrival model with the fully online model. In the fully online matching problem, vertices arrive and depart online. Upon the arrival of a vertex, its incident edges to existing vertices are revealed. This is identical to the setup in the two-sided model. However, in the fully online model, matching decisions can be deferred until the departure of a vertex, which represents its final opportunity to be matched. The two-sided model is not easier than the fully online model, as any feasible algorithm for the former can be directly applied to the latter, achieving the same competitive ratio. Although prior work has established separation results between these models, our findings underscore the substantial loss incurred when matching decisions must be made upon arrival rather than at departure. Notably, the fully online model admits a $0.6$-competitive fractional algorithm~\cite{soda/HuangPTTWZ19,focs/HuangTWZ20,ec/TangZ24}, as well as $0.569$-competitive~\cite{focs/HuangTWZ20} and $0.521$-competitive~\cite{jacm/HuangKTWZZ20} integral algorithms for bipartite and general graphs, respectively.

\subsection{Related Work}
There is a huge literature in online matching, refer to the impactful survey by Mehta~\cite{fttcs/Mehta13}, and a more recent one by Huang, Tang, and Wajc~\cite{ec/HuangTW24}.

Despite the fact that fractional algorithms are stronger than integral ones, fractional algorithms have immediate algorithmic implications when every vertex has a large capacity, e.g., the $b$-matching setting~\cite{tcs/KalyanasundaramP00} in one-sided arrival model, and the small-bid setting~\cite{jacm/MehtaSVV07} in AdWords. Furthermore, researchers have studied explicitly how to round fractional algorithms in online matching~\cite{soda/CohenW18,soda/BuchbinderNW23,stoc/TangWW22}.

The two-sided/general vertex arrivals model has also been applied in the prophet setting~\cite{mor/EzraFGT22} and the secretary setting~\cite{ec/EzraFGT22}. In stark contrast to our adversarial setting, the general vertex arrival model admits the same optimal $0.5$ competitive ratio in the prophet setting as in the one-sided model~\cite{soda/FeldmanGL15}, and admits a larger competitive ratio of $\nicefrac{5}{12}$ in the secretary setting than the optimal $\nicefrac{1}{e}$ ratio in the one-sided model~\cite{esa/KesselheimRTV13}.

\section{Preliminaries}
\label{sec:prelim}
We are interested in the fractional version of online bipartite matching with two-sided vertex arrivals.
Fix an arbitrary algorithm, we use $x_{uv} \in [0,1]$ to denote the fractional matching selected by the algorithm of each edge $(u,v)$, and use $x_v = \sum_{(u,v)\in E} x_{uv}$ to denote the total matched portion of each vertex $v$.

Recall that the only restriction on an algorithm is $x_v \le 1$ for each vertex $v$. It suffices to study deterministic fractional algorithms, since any randomized one can be converted by assigning each edge $(u,v)$ its expected matched portion under the randomized algorithm. The resulting deterministic algorithm remains feasible and yields the same expected matching size.
\paragraph{Intrinsic Difficulty of Two-sided Arrivals.}
In both the one-sided model and the fully online model, an important feature is that the algorithm is informed when all edges of a vertex have arrived. Specifically, upon the arrival of an online vertex in the one-sided arrival model, or upon the departure of a vertex in the fully online model, this information becomes available. It is then without loss of optimality for the algorithm to match this vertex greedily.

By contrast, the fundamental difficulty of the two-sided arrival model lies in the absence of such information. As a simple toy example, consider a single edge between the first two vertices. If this is the only edge in the graph, the algorithm should match it. However, another possibility is that two additional vertices arrive later, each connecting to one of the first two vertices. It is straightforward to verify that this instance yields an upper bound of $\tfrac{2}{3}$ on the competitive ratio.

Notice, moreover, that upon the arrivals of the third and fourth vertices, the algorithm again faces the same uncertainty as in the first step. This situation can continue to repeat. Our construction in Section~\ref{sec:construction} captures and magnifies this difficulty to the extreme.

\paragraph{Tang and Zhang's Algorithm~\cite{ec/TangZ24}.}
As mentioned in the introduction, two existing fractional algorithms achieve the same $\Gamma^*$ competitive ratio. Here we present the algorithm of Tang and Zhang~\cite{ec/TangZ24}, since an optimization problem they solved also appears in our analysis, which causes our upper bound on the competitive ratio to coincide with their algorithmic result.

Their algorithm is specified by a non-increasing function $a: [0,1] \to [0,1]$, and operates as follows:
\begin{itemize}
\item Upon the arrival of a vertex $u$, fractionally and continuously match it to the least matched neighbor $v$ until $x_u = a(x_v)$ or $x_v = 1$, where $a$ is the prefixed non-increasing function.
\end{itemize}

Intuitively, the function $a$ is optimized to balance the two cases discussed earlier. Their primal–dual analysis introduces an auxiliary function $g$. We restate their result in a form convenient for our purposes and omit the proof, which is not essential for our work.

\begin{fact}
\label{fact:tz}
If there exist a function $a:[0,1] \to [0,1]$, a non-decreasing function $g:[0,1]\to [0,1]$, and a constant $\Gamma$ such that:
\begin{itemize}
\item $a(x) (1-g(x)) \ge \int_0^{a(x)} g(y) \dd y$,
\item $a(x) (1-g(x)) + \int_0^x g(y) \dd y \ge \Gamma$,
\end{itemize}
then the algorithm is $\Gamma$-competitive.
\end{fact}

Based on this fact, it suffices to design functions $a,g$ so that the corresponding ratio $\Gamma$ is maximized. Tang and Zhang explicitly constructed functions $a,g$ attaining the $\Gamma^*$ competitive ratio, though without proving the optimality of their construction. As one might expect, these inequalities also appear in our analysis after appropriate transformations, and in Section~\ref{sec:math} we prove the optimality of their construction.
\paragraph{Roadmap.} In Section~\ref{sec:math}, we establish a sequence of mathematical facts, that are directly related to the functions $a,g$ of Tang and Zhang~\cite{ec/TangZ24}. In Section~\ref{sec:construction}, we provide our construction. In Section~\ref{sec:main}, we prove the main theorem. 

\section{Mathematical Facts}
\label{sec:math}

Fix two constants $\eps, \Gamma > 0$. Readers may think of $\eps$ as a small positive number, and $\Gamma$ as a constant slightly larger than the constant $\Gamma^*$, as defined in Theorem~\ref{thm:main}. In Section~\ref{sec:construction}, we will further restrict $\varepsilon$ to be the reciprocal of an integer, though this constraint is not relevant in the current section.

We introduce the following family of functions $\{F_i\}_{i=1}^{\infty}$ defined on the interval $[0,1]$. The functions are defined recursively as follows:
\begin{align*}
& F_1(x):= 1-\frac{x}{2}-\Gamma~; \\
& F_{n+1}(x) := \max_{ \left\{ \substack{0\le a\le 1, \\ x+\eps a\le 1} \right\}} \min \left\{ F_n(x+\eps a) + \eps F_n(a), (1-\eps) F_n(x+\eps a)+ \eps \left(\frac{(1+\eps)a + x}{2} - \Gamma \right) \right\}~.
\end{align*}

In Section~\ref{sec:construction}, we will present a recursive construction of the instance over $n$ steps. The recursive definitions of the above functions are closely aligned with this construction. In particular, the value $F_n(0)$ represents the difference between the size of the matching produced by the algorithm and $\Gamma$ times the size of the optimal matching. As we will show in the next section, if $F_n(0) < 0$, then the algorithm does not achieve a competitive ratio better than $\Gamma$.

In this section, we prove that if $\Gamma > \Gamma^*+\eps$, $F_n(0)$ would eventually be negative for sufficiently large number $n$. This is a purely technical statement. The connection between these functions and the algorithmic result of Tang and Zhang~\cite{ec/TangZ24} will become clear in the analysis. 

We start with a few basic properties of the function $F_n$. We first prove that for every fixed $x$, the sequence $\{F_n(x)\}_n$ is non-increasing.
\begin{claim}
\label{cl:monotone}
For every $n$, it holds that $F_n(x) \ge F_{n+1}(x), \forall x \in [0,1]$.
\end{claim}
\begin{proof}
We prove the statement by induction on $n$. We first prove the base case that $\forall x \in [0,1], F_1(x) \ge F_2(x)$:
\begin{align*}
F_2(x) & \le \max_{ \left\{ \substack{0\le a\le 1, \\ x+\eps a\le 1} \right\}} \left( (1-\eps) F_1(x+\eps a)+ \eps \left(\frac{(1+\eps)a + x}{2} - \Gamma \right) \right) \\
& = \max_{ \left\{ \substack{0\le a\le 1, \\ x+\eps a\le 1} \right\}}\left( (1-\eps) \left( 1- \frac{x+\eps a}{2} - \Gamma \right) + \eps \left(\frac{(1+\eps)a + x}{2} - \Gamma \right) \right) \\
& = \max_{ \left\{ \substack{0\le a\le 1, \\ x+\eps a\le 1} \right\}} \left( 1 - \frac{x}{2} - \Gamma + \eps ( x+\eps a-1) \right) \\
& \le 1-\frac{x}{2} - \Gamma = F_1(x)~,
\end{align*}
where the first inequality holds by dropping the first term of the minimum from the definition, and the last inequality holds by the fact that $x+\eps a \le 1$. This finishes the proof of the base case.

Suppose the statement holds for $n$, i.e., $F_n(x) \le F_{n-1}(x), \forall x \in [0,1]$. Then we have,
\begin{align*}
F_{n+1}(x) & = \max_{ \left\{ \substack{0\le a\le 1, \\ x+\eps a\le 1} \right\}} \min \left\{ F_n(x+\eps a) + \eps F_n(a), (1-\eps) F_n(x+\eps a)+ \eps \left(\frac{(1+\eps)a + x}{2} - \Gamma \right) \right\} \\
& \le \max_{ \left\{ \substack{0\le a\le 1, \\ x+\eps a\le 1} \right\}} \min \left\{ F_{n-1}(x+\eps a) + \eps F_{n-1}(a), (1-\eps) F_{n-1}(x+\eps a)+ \eps \left(\frac{(1+\eps)a + x}{2} - \Gamma \right) \right\} \\
& = F_n(x)~,
\end{align*}
where the inequality holds by inductive hypothesis.
\end{proof}

Next, we prove the function is concave.
\begin{claim}
\label{cl:concave}
For every $n$, $F_n(\cdot)$ is concave.
\end{claim}
\begin{proof}
We prove the statement by induction. The base case when $n=1$ holds by inspection to the definition. Particularly, $F_1(\cdot)$ is a linear function.
Suppose the statement holds for $n$. For an arbitrary $x',x'' \in [0,1]$, it suffices to prove that for every $t \in [0,1]$ and $x=t x'+ (1-t)x''$, we have $F_{n+1}(x) \ge t F_{n+1}(x')+ (1-t)F_{n+1}(x'')$. Let $a',a''$ be the arguments of the maxima for $F_{n+1}(x'), F_{n+1}(x'')$ respectively. Let $a=t a'+ (1-t)a''$. Then we have,
\begin{align*}
F_{n+1}(x) & \ge \min \left\{F_n(x+\eps a)+\eps F_n(a), (1-\eps) F_n(x+\eps a) + \eps \left(\frac{(1+\eps)a + x}{2} - \Gamma \right) \right\} \\
& \ge \min \bigg\{ \left(t F_n(x'+\eps a') + (1-t)F_n(x''+\eps a'')\right) +  \eps \left(tF_n(a') + (1-t)F_n(a'')\right) , \\
& \quad \left(tF_n(x'+\eps a')+(1-t)F_n(x''+\eps a'') \right) + \eps \left(\frac{(t((1+\eps)a'+x') + (1-t)((1+\eps)a''+x'')}{2} - \Gamma \right) \bigg\} \\
& \ge tF_{n+1}(x') + (1-t)F_{n+1}(x'')~,
\end{align*}
where the second inequality holds by the inductive hypothesis, and the last inequality holds by the fact that $\min(b_1+b_2,c_1+c_2) \ge \min(b_1,c_1) + \min(b_2,c_2)$. This concludes the proof of the statement.
\end{proof}

Finally, we prove the function is $\frac{1}{2}$-Lipschitz.
\begin{claim}
\label{cl:lipschitz}
For every $n$, $F_n(\cdot)$ is $\frac{1}{2}$-Lipschitz, i.e., $|F_n(x)-F_n(x')| \le \frac{1}{2} |x-x'|$.
\end{claim}
\begin{proof} 
We prove the statement by induction. The base case when $n=1$ holds by inspection to the definition. 
Suppose the statement holds for $n$. For arbitrary $x,x' \in [0,1]$, let $a'$ be the argument of the maxima $F_{n+1}(x')$.

Let $a \in [0,1]$ be that $x+\eps a = \min(x+\eps a',1)$. Then we have,
\begin{align}
F_{n+1}(x) & \ge \min \left\{F_n(x+\eps a)+\eps F_n(a),  (1-\eps) F_n(x+\eps a) + \eps\left(\frac{(1+\eps)a + x}{2} - \Gamma \right) \right\} \notag \\
& \ge \min \bigg\{ F_n(x'+\eps a')+\eps F_n(a') - \frac{1}{2} |x+\eps a-x'-\eps a'| - \frac{1}{2} \eps |a-a'| , \notag \\
&  (1-\eps)F_n(x'+\eps a') + \eps \left(\frac{(1+\eps)a' + x'}{2} - \Gamma \right) - \frac{1}{2} |x+\eps a-x'-\eps a'| - \frac{1}{2} \eps|a-a'| \bigg\}, 
\label{eqn:lipschitz}
\end{align}
where the last inequality holds by the inductive hypothesis. Furthermore, we verify that
\[
|x+\eps a-x'-\eps a'|+\eps|a-a'| = |x-x'|~.
\]
We consider the following two cases:
\begin{itemize}
\item If $x+\eps a' \le 1$, then we have $a=a'$.	 Consequently, $|x+\eps a-x'-\eps a'| + \eps |a-a'| = |x-x'|$.
\item If $x+\eps a' > 1$, then it must be the case that $x>x'$ and $a<a'$, since $x'+\eps a'\le 1$. Consequently,
\[
|x+\eps a-x'-\eps a'| + \eps |a-a'| = x+\eps a - x'-\eps a' + \eps (a'-a) = x-x'.
\]
\end{itemize}
Substituting the equation to \eqref{eqn:lipschitz}, we have that
\begin{multline*}
F_{n+1}(x) \ge \min \Big\{ F_n(x'+\eps a')+\eps F_n(a') - \frac{1}{2} |x-x'|, \\
 (1-\eps)F_n(x'+\eps a') + \eps \left(\frac{(1+\eps)a' + x'}{2} - \Gamma \right) - \frac{1}{2} |x-x'|\Big\} = F_{n+1}(x')-\frac{1}{2} |x-x'|~,
\end{multline*}
which concludes  the proof of the claim.
\end{proof}

Now, we are ready to prove our main lemma.
\begin{lemma}
\label{lem:math}
If $\Gamma > \Gamma^* + \eps$, there exists a sufficiently large $n$ so that $F_n(0) < 0$.	
\end{lemma}
\begin{proof}
We prove the statement by contradiction. Suppose $F_n(0) \ge 0$ for every $n$. 
By the Lipschitzness of $F_n$, we have that $F_n(x) \ge F_n(0) - \frac{x}{2} \ge - \frac{x}{2}$. 

Thus, for every $x \in [0,1]$, $\{F_n(x)\}_n$ is a non-increasing sequence (by Claim~\ref{cl:monotone}) and is lower bounded by $-\frac{x}{2}$. Hence, we define $F(x) := \lim_{n \to \infty}F_n(x)$.
Moreover, this limit function $F(\cdot)$ is also $\frac{1}{2}$-Lipschitz and concave. 
In addition, by taking the limit of $F_n$ in its definition, the function $F$ satisfies the following property:\begin{align*}
\forall x \in [0,1], F(x) = \max_{ \left\{ \substack{0\le a\le 1, \\ x+\eps a\le 1} \right\}} \min \left\{ F(x+\eps a) + \eps F(a), (1-\eps)F(x+\eps a) + \eps \left(\frac{(1+\eps)a + x}{2} - \Gamma \right)\right\}~.
\end{align*}
Therefore, there exists a corresponding function $a:[0,1] \to [0,1]$ such that
\begin{align*}
\forall x \in [0,1], \quad & F(x+\eps a(x)) +  \eps F(a(x)) \ge F(x)~, \\
& (1-\eps)F(x+\eps a(x)) + \eps \left(\frac{(1+\eps)a(x)+x}{2}-\Gamma \right) \ge F(x)~.
\end{align*}
Let $G(x) := \frac{x}{2} - F(x)$, and we rearrange the above two inequalities as follows:
\begin{align*}
\forall x \in [0,1], \quad & G(x+\eps a(x)) - G(x) \le \eps (a(x) - G(a(x)))\\
& G(x+\eps a(x)) - G(x) \le \eps (a(x)+G(x+\eps a(x))-\Gamma)~.
\end{align*}
Moreover, since $F(0) \ge 0$ and $F(1) \le F_1(1) = \frac{1}{2} -\Gamma$, we have that $G(0) \le 0$ and $G(1) \ge \Gamma$.

Since $F(\cdot)$ is $\frac{1}{2}$-Lipschitz and concave, we have that $G(\cdot)$ is non-decreasing, $1$-Lipschitz, and convex.
Therefore, $G$ is absolutely continuous and admits an almost-everywhere derivative. In particular, there exists a non-decreasing function $g : [0,1] \to [0,1]$ such that
\[
G(x) = G(0) + \int_0^x g(y) \dd y~, \quad \forall x \in [0,1]~.
\]
For convenience of notation, we will refer to this function $g$ as the ``derivative'' of $G$ throughout the proof.
Here, $g$ is upper bounded by $1$ since $G$ is $1$-Lipschitz.
Next, we have that
\[
G(x+\eps a(x))-G(x) \ge \eps a(x) \cdot g(x)~. 
\]
Applying the inequality to the above two inequalities and eliminating $\eps$ from both sides, we have that
\begin{align*}
\forall x \in [0,1], \quad & a(x) \cdot (1-g(x)) \ge G(a(x)) \\
& a(x) \cdot (1-g(x))+ G(x+\eps a(x)) \ge \Gamma~.
\end{align*}
We further apply the inequality $G(x+\eps a(x)) \le G(x)+\eps a(x) \le G(x) + \eps$ to the second condition.
To sum up, there exists a function $a:[0,1] \to [0,1]$ such that
\begin{align*}
\forall x \in [0,1], \quad & a(x) \cdot (1-g(x)) \ge G(a(x)) ~,\\
& a(x) \cdot (1-g(x))+G(x) \ge \Gamma - \eps~.
\end{align*}

We remark that the above family of inequalities are essentially the same properties desired by Tang and Zhang~\cite{ec/TangZ24} as explained in Section~\ref{sec:prelim}. We prove the optimality of their construction in the following lemma, whose proof is deferred to the end of the section.
\begin{lemma}
\label{lem:optimal}
If there exists a function $a:[0,1] \to [0,1]$, a non-decreasing function $g:[0,1] \to [0,1]$, and a convex function $G(x):= G(0) + \int_0^x g(y) \dd y$, such that
\begin{itemize}
\item $a(x) (1-g(x)) \ge G(a(x))$~,
\item $a(x) (1-g(x)) + G(x) \ge \tilde\Gamma$~,
\item $G(0) \le 0$ and $G(1) \ge \tilde{\Gamma}$~,
\end{itemize}
then we have $\tilde{\Gamma} \le \Gamma^*$.
\end{lemma}
Applying the lemma to the functions $a,g,G$ we derived, and the constant $\tilde\Gamma=\Gamma-\eps$, we have $\Gamma-\eps \le \Gamma^*$, which contradicts the assumption that $\Gamma > \Gamma^* + \eps$.
\end{proof}

\subsection{Proof of Lemma~\ref{lem:optimal}}
We first fix the function $g$ and restrict ourselves to $X=\{x: G(x) \in [0, \tilde{\Gamma}]\}$. If there exists $x\in X$ with $g(x)=1$, by the second condition, we must have $G(x) \ge \tilde\Gamma$, and hence $G(x)=\tilde\Gamma$. In this case, we define $\frac{\tilde\Gamma-G(x)}{1-g(x)} := 0$.
Consequently, for every $x \in X$, we have that $a(x) \ge \frac{\tilde\Gamma-G(x)}{1-g(x)}$.

We further consider the function $l(a):=a(1-g(x))-G(a)$. This function is concave in $a$, since its derivative $l'(a)=1-g(x)-g(a)$ is non-increasing. 
Therefore, by the concavity of $l$ and the fact that $\frac{\tilde\Gamma-G(x)}{1-g(x)} \in [0,a(x)]$, we have
\begin{align}
\label{eqn:condition}
& l\left(\frac{\tilde\Gamma-G(x)}{1-g(x)}\right) \ge \min (l(a(x)), l(0)) =0 \notag \\
\Longrightarrow & \ \frac{\tilde\Gamma-G(x)}{1-g(x)} \cdot (1-g(x)) - G\left(\frac{\tilde\Gamma-G(x)}{1-g(x)}\right) \ge 0 \notag \\
\Longrightarrow & \ \tilde\Gamma -G(x) \ge G\left( \frac{\tilde\Gamma-G(x)}{1-g(x)} \right).
\end{align}

Recall that $G(0) \le 0$ and $G(1) \ge \tilde{\Gamma}$. Let $H:[0,\tilde{\Gamma}] \to [0,1]$ be defined as $H(y):=G^{-1}(y) = \sup\{x:G(x)\le y\}$. 
Then, \eqref{eqn:condition} implies that
\[
H(\tilde\Gamma-y) \ge \frac{\tilde\Gamma-y}{1-\frac{1}{H'(y)}}~, \quad \forall y \in [0,\tilde\Gamma]~.
\]
Here we use the fact that $H$ is non-decreasing and the inverse function rule that $H'(y)=\frac{1}{G'(x)}= \frac{1}{g(x)} \ge 1$. Rearranging the inequality gives that
\[
H'(y) \ge \frac{H(\tilde\Gamma- y)}{H(\tilde\Gamma - y) - (\tilde\Gamma - y)}~, \quad \forall y \in [0,\tilde\Gamma]~.
\]
Integrating the above equation from $y$ to $\tilde\Gamma$ gives the following:
\begin{equation}
\label{eqn:H}
1-H(y) \ge H(\tilde\Gamma) - H(y) =\int_y^{\tilde\Gamma} H'(z) \dd z \ge \int_y^{\tilde\Gamma} \frac{H(\tilde\Gamma- z)}{H(\tilde\Gamma - z) - (\tilde\Gamma - z)} \dd z~, \quad \forall y \in [0,\tilde\Gamma]~.
\end{equation}
Moreover, $H(y) = H(0) + \int_0^y H'(z) \dd z \ge 0 + \int_0^y 1 \dd z = y$.
The next claim states that if such a function $H$ exists, there must exist $\widehat{H}$ that attains equality for every $y$. 
The proof is similar to that of Lemma 4 of Wang and Wong~\cite{icalp/WangW15}.
\begin{claim}
If $H:[0,r] \to [0,1]$ satisfies that $H(y) \ge y$, and
\[
1-H(y) \ge \int_y^{r} \frac{H(r- z)}{H(r - z) - (r - z)} \dd z~, \quad \forall y \in [0,r]~.
\]
then there exists $\widehat{H}:[0,r] \to [0,1]$ such that
\[
1-\widehat{H}(y) = \int_y^{r} \frac{\widehat{H}(r-z)}{\widehat{H}(r-z)-(r-z)} \dd z~, \quad \forall y \in [0,r]~.
\]
\end{claim}
\begin{proof}
Let $H_0=H$. We recursively construct a family of functions: 
\begin{equation}
\label{eqn:H1}
H_{i+1}(y) = 1-\int_y^r \frac{H_i(r-z)}{H_i(r-z)-(r-z)} \dd z~.
\end{equation}
We prove by induction that $H_{i+1}(y) \ge H_i(y), \forall y \in [0,r]$. The statement holds for $i=0$ by the condition of the claim. If the statement holds for $i-1$, we then have
\begin{align*}
H_{i+1}(y) & = 1-\int_y^r \frac{H_i(r-z)}{H_i(r-z)-(r-z)} \dd z \\
& \ge 1-\int_y^r \frac{H_{i-1}(r-z)}{H_{i-1}(r-z)-(r-z)} \dd z = H_i(y)~.
\end{align*}
where the inequality holds by induction hypothesis. Therefore, $\{H_i(y)\}_i$ is non-decreasing and bounded above by $1$ for every $y$. Finally, we define the limit function 
\[
\widehat{H}(y) := \lim_{i \to \infty} H_i(y)~, \quad \forall y \in[0,r]~.
\]
The function $\widehat{H}$ satisfies the condition of the claim by taking the limit of \eqref{eqn:H1} on both sides.
\end{proof}

Finally, without loss of generality, we assume $H$ achieves equality in \eqref{eqn:H} for every $y \in [0,\tilde{\Gamma}]$, i.e., 
\[
1-H(y) = \int_y^{\tilde\Gamma} \frac{H(\tilde\Gamma-z)}{H(\tilde\Gamma-z)-(\tilde\Gamma-z)} \dd z~, \quad \forall y \in [0,\tilde{\Gamma}]~.
\]
Taking derivative on the above equation gives that
\begin{equation}
H'(y) = \frac{H(\tilde\Gamma- y)}{H(\tilde\Gamma - y) - (\tilde\Gamma - y)}~, \quad \forall y \in [0,\tilde\Gamma]~.
\label{eqn:H_equal}
\end{equation}
We proceed by solving the function explicitly. Observe that
\begin{align*}
& ((H(y)-y)\cdot (H(\tilde{\Gamma}-y)-(\tilde{\Gamma}-y)))' \\
= & \ (H'(y)-1)\cdot(H(\tilde{\Gamma}-y)-(\tilde{\Gamma}-y)) - (H'(\tilde{\Gamma}-y)-1)\cdot (H(y)-y) \\
= & \ \frac{\tilde{\Gamma}-y}{H(\tilde{\Gamma}-y)-(\tilde{\Gamma}-y)}\cdot(H(\tilde{\Gamma}-y)-(\tilde{\Gamma}-y)) - \frac{y}{H(y)-y} \cdot (H(y)-y) = \tilde{\Gamma} -2y~.
\end{align*}
Let $c=H(0)$. Integrating the above equation from $0$ to $y$, we have
\[
(H(y)-y)\cdot (H(\tilde{\Gamma}-y)-(\tilde{\Gamma}-y)) = H(0)\cdot (H(\tilde{\Gamma})-\tilde{\Gamma}) + \tilde{\Gamma} y - y^2 = c(1-\tilde{\Gamma})+\tilde{\Gamma} y-y^2~,
\]
where the last equality holds by $H(\tilde{\Gamma}) = 1$. Consequently, we have
\begin{equation}
(\ln(H(y)-y))' = \frac{H'(y)-1}{H(y)-y} = \frac{(H'(y)-1) \cdot (H(\tilde{\Gamma}-y)-(\tilde{\Gamma} - y))}{(H(y)-y) \cdot (H(\tilde{\Gamma}-y)-(\tilde{\Gamma}-y))} =\frac{\tilde{\Gamma}-y}{c(1-\tilde{\Gamma}) + \tilde{\Gamma} y - y^2}~.
\label{eqn:ln_H}
\end{equation}
We omit the tedious calculation and express the right hand side as follows:
\[
\frac{\tilde{\Gamma}-y}{c(1-\tilde{\Gamma}) + \tilde{\Gamma} y - y^2} = \frac{\alpha_1}{y-r_1} + \frac{\alpha_2}{y-r_2}~,
\]
where 
\[
k=\frac{\sqrt{\tilde{\Gamma}^2+4c(1-\tilde{\Gamma})}}{\tilde{\Gamma}} \ge 1, 
r_{1}=\tilde{\Gamma} \cdot \frac{1 + k}{2}, \alpha_{1}=\frac{k -1}{2k}, r_{2}=\tilde{\Gamma} \cdot \frac{1 - k}{2}, \alpha_{2}=\frac{k+1}{2k}~.
\]
Furthermore, we rewrite $c$ as follows:
\[
c=\frac{(k^2-1)\tilde{\Gamma}^2}{4(1-\tilde{\Gamma})}=\frac{-r_1r_2}{1-\tilde{\Gamma}}~.
\]
Integrating \eqref{eqn:ln_H} from $0$ to $\tilde\Gamma$ gives that
\begin{multline*}
\ln(H(\tilde\Gamma)-\tilde\Gamma)-\ln(H(0)-0) \ge \int_0^{\tilde\Gamma} \frac{\alpha_1}{y-r_1}  + \frac{\alpha_2}{y-r_2} \dd y \\
= \alpha_1\ln \left(\frac{r_1-\tilde\Gamma}{r_1}\right) +\alpha_2 \ln\left(\frac{\tilde\Gamma-r_2}{-r_2}\right) = \alpha_1 \ln \left( \frac{-r_2}{r_1} \right) + \alpha_2 \ln \left( \frac{r_1}{-r_2} \right)~.
\end{multline*}
Finally, we have that
\[
1 \ge H(\tilde{\Gamma})\ge \tilde{\Gamma} + c \cdot \left( \frac{-r_2}{r_1}\right)^{\alpha_1}\cdot \left(\frac{r_1}{-r_2}\right)^{\alpha_2} = \tilde{\Gamma}+\frac{(r_1)^{2\alpha_2}\cdot (-r_2)^{2\alpha_1}}{1-\tilde{\Gamma}}~.
\]
By rearranging the inequality and substituting $r_{1,2}$ and $\alpha_{1,2}$ in the format of $k$, we have that
\[
\tilde{\Gamma} \le \frac{1}{\left( \frac{k+1}{2}\right)^{\frac{k+1}{2k}} \left( \frac{k-1}{2}\right)^{\frac{k-1}{2k}}+1} \le \Gamma^*~,
\]
where the last inequality holds by the definition of $\Gamma^*$. This concludes the proof of the lemma.
 
\section{Main Construction}
\label{sec:construction}
For ease of presentation, we allow multiple vertices to arrive simultaneously. It is guaranteed that no edges exist between the newly arrived vertices. The edges connecting these new vertices to the existing ones are revealed at the same time, and the algorithm makes matching decisions for all newly revealed edges in a single step. Note that this generalization can be simulated in the original model by having the vertices arrive sequentially, one at a time.
Importantly, allowing simultaneous arrivals only makes the algorithm's task easier.

Recall that we use $x_v$ to denote the matched portion of a vertex $v$. For an arbitrary set of vertices $S$, we define $x_S$ to denote the average matched portion of the vertices in $S$, i.e., $x_S := \frac{1}{|S|} \cdot \sum_{v \in S} x_v$.

We introduce a slightly more general version of the online matching problem that involves an extra initialization step. Fix an $\eps >0$ so that $\eps^{-1}$ is an integer. Our construction is parameterized by a three tuple $(n,N,x)$, where $n$ is the number of steps of our construction, $N$ is a multiple of $\eps^{-n}$, and $x \in [0,1]$ corresponds to the initial matched level that we explain below.

\paragraph{Initialization.} 
The instance starts with an empty graph $(V_0,E_0=\emptyset)$ with $|V_0|=N$ isolated vertices. 
Next, the algorithm selects a vector $\vect{x} \in [0,1]^{V_0}$, so that $\sum_{v \in V_0} x_v = x \cdot |V_0|$. We interpret the number $x_v \in [0,1]$ as the matched portion of vertex $v$, though there are no edges in the initial graph. Notice that we allow the algorithm to distribute the matched portion arbitrarily, provided that the average matched portion of $V_0$ equals the prefixed number $x$.

Note that having an initial average matched portion $x$ is not necessarily advantageous for the algorithm. Consider the extreme case where $x = 1$: if $N$ additional vertices arrive in the next step and are fully connected to the existing vertices, the initial matching prevents the algorithm from matching any of the new edges, resulting in a competitive ratio of $0.5$. In our recursive construction, the initial value of $x$ is always inherited from the algorithm’s matching decisions in previous steps.

Observe that the original matching problem corresponds to the case $N = 0$, where no vertices exist initially. Furthermore, any instance of the form $(n, N, x = 0)$ can be simulated in the original model by having $N$ isolated vertices arrive one by one, so the algorithm has no matching decisions to make during this phase.

\paragraph{Invariants.}
The adversary maintains a partition of the vertices of the graph:
\[
V = A \cup I = \left(\cup_{i} A_i \right) \cup I~,
\]
where $A$ and $I$ are referred as \emph{active} and \emph{inactive} vertices respectively.  
The active vertices are further partitioned into $\cup A_i$ so that the subgraph induced by each $A_i$ is an empty graph. 

Moreover, the vertices of each $A_i$ are \emph{symmetric} in the sense that they share the same set of neighbors, rendering them indistinguishable to the algorithm. While this symmetry is not explicitly utilized in our analysis, it justifies why a single value representing the average matched portion suffices to capture the matching status.

Vertices are deactivated by the adversary in pairs of neighbors, ensuring that the subgraph induced by $I$ admits a perfect matching. Moreover, once a vertex becomes inactive, it shall no longer have further edges. By the end of the instance, i.e., after $n$ steps, all vertices of the graph become inactive.

At the beginning, all vertices in $V_0$ are active and belong to the same partition. 

\paragraph{Base Construction.} Suppose we reach the last step of the construction. For each active partition $A_i$, let a set of vertices $B_i$ arrive, where $(A_i,B_i)$ forms a complete bipartite graph, and $|B_i|=|A_i|$. 

The algorithm makes matching decisions between $A_i$ and $B_i$. Afterwards, we deactivate all vertices in $A_i \cup B_i$.
We apply the construction to every active partition and then ends the instance.

\paragraph{Recursive Construction.}
Suppose we have finished $n-k$ steps and have $k$ more steps with $k>1$. For each active partition $A_i$ (if exists), let a set of vertices $B_i$ arrive, where $(A_i,B_i)$ forms a complete bipartite graph, and $|B_i| = \eps |A_i|$. 

The algorithm then makes matching decisions between $A_i$ and $B_i$.
Suppose the algorithm selects a total of $a \cdot |B_i|$ fractional edges between $A_i$ and $B_i$. 
Then, the average matched portion $x_{B_i}$ increases by $a$, and the average matched portion $x_{A_i}$ increases by $\eps a$. In the next equation, we use $x_{A_i}$ to denote the average matched portion of $A_i$ \emph{before} the arrival of $B_i$.
 
We update the partition of vertices based on the value of $a$. Specifically, we compare the value of the following two quantities:
 \[
F_{k-1}(x_{A_i}+\eps a) + \eps F_{k-1}(a), \quad \text{and} \quad (1-\eps) F_{k-1}(x_{A_i}+\eps a)+ \eps \left(\frac{(1+\eps)a + x_{A_i}}{2} - \Gamma \right)~.
 \]
\begin{itemize}
\item {\bf Against Aggressive Algorithms:}
If the first term is smaller or equal to the second term, we keep the set $A_i$ active, update $B_i$ as active vertices and group $B_i$ as an independent active partition.

Intuitively, this corresponds to the case when the algorithm matches too many edges between $A_i$ and $B_i$ (i.e., $a$ is large). We shall have the perfect neighbors of $A_i$ and $B_i$ arrive later in our construction, so that the selected matchings between $A_i$ and $B_i$ are bad matching decisions.

\item {\bf Against Conservative Algorithms:} 
If the second term is smaller, we deactivate all vertices of $B_i$, and the $|B_i|$ \emph{least matched} vertices of $A_i$, which we denote as $C_{i}$. I.e., we add $B_i \cup C_i$ to the set of inactive vertices $I$.
Then, we keep the set $A_i \setminus C_i$ as a partition of active vertices.

Intuitively, this corresponds to the case when the algorithm matches too less edges between $A_i$ and $B_i$ (i.e., $a$ is small).
Notice that $B_i \cup C_i$ admits a perfect matching as they form a complete bipartite graph with the same number of vertices on both sides. However, the algorithm failed to match these edges sufficiently. 
\end{itemize}
We apply the construction to every active partition and then recurse.

\subsection{Basic Properties}
We first list a few basic properties of the instance.
The first observation is that throughout our construction, the size of $|B_i| = \eps|A_i|$ is an integer, so that the above description of our instance is valid.
\begin{claim}
If the remaining number of steps is $k$, then the number of vertices of each active partition is a multiple of $\eps^{-k}$.
\end{claim}
\begin{proof}
We prove the statement by backward induction. The base case is when $k=n$, the statement holds by our assumption of $N$. Suppose the statement holds for $k$. There would be no active vertices after the last step. Hence, we only consider the case when $k\ge 2$. Each active partition $A_i$ is associated with a set of vertices $B_i$ arriving online, with $|B_i|=\eps |A_i|$. Notice that we either add $B_i$ as a new active partition, in which case the number of vertices $|B_i|$ is a multiple of $\eps^{-(k-1)}$ since $|A_i|$ is a multiple of $\eps^{-k}$ by induction hypothesis; or we remove $C_i$ from $A_i$, in which case the number of vertices $|A_i \setminus C_i| = |A_i|-|C_i|=|A_i|-|B_i|$ is also a multiple of $\eps^{-(k-1)}$. This concludes the proof of the claim. 
\end{proof}

The second observation is that our construction provides a bipartite graph.
\begin{claim}
The instance is a bipartite graph.
\end{claim}
\begin{proof}
It suffices to provide a two coloring (white and black) of the vertices at their arrival. At the beginning, we color all vertices of $V_0$ in white. We maintain the invariant that all vertices of an active partition $A_i$ have the same color. At each step, every active partition $A_i$ is associated with a set of vertices $B_i$ arriving online. Since $A_i$ and $B_i$ form a complete bipartite graph, we color the vertices in $B_i$ using the different color of $A_i$. Notice that we either add $B_i$ as a new active partition, or deactivate a subset of vertices from $A_i$. The invariant holds after each step, which concludes the proof.  
\end{proof}

The third observation is that our graph admits a perfect matching.
\begin{claim}
\label{cl:perfect}
There exists a perfect matching in the constructed graph.
\end{claim}
\begin{proof}
We prove that throughout the construction, there exists a perfect matching in the subgraph induced by $I$.  Indeed, our construction deactivates vertices in pairs: either $A_i \cup B_i$ in the last step, or $B_i \cup C_i$ in the recursive steps. At the end, all vertices are inactive, and hence, there exists a perfect matching in the constructed graph after $n$ steps.
\end{proof}

\subsection{Competitive Analysis}
We have specified the behavior of the adversary. Thus, for an arbitrary algorithm $\alg$, the resulting instance is captured by the three parameters $(n,N,x)$.
Let $I$ be the final set of vertices of the graph of $\alg$. Formally, $I$ is a function of $(\alg,n,N,x)$. We omit the dependency for notation simplicity.

Then we have that the optimal matching has size $\frac{1}{2}\cdot |I|$ by Claim~\ref{cl:perfect}.
And the total matching produced by the algorithm (including the $xN$ fractions from the initialization step) is $\frac{1}{2} \cdot \sum_{v \in I} x_v$, where $x_v \in [0,1]$ is the total matched portion of $v$.
We use $v_\alg(n,N,x)$ to denote the difference between the size of the matching produced by the algorithm, and $\Gamma$ times the size of the maximum matching:
\[
v_\alg(n,N,x) := \frac{1}{2} \cdot \left( \sum_{v \in I} x_v - \Gamma \cdot |I| \right)~.
\]
Notice that if $\alg$ is $\Gamma$-competitive, it should have the property that $v_\alg(n,N,0) \ge 0$.
Next, we prove the main lemma.
\begin{lemma}
\label{lem:construction}
For any algorithm $\alg$, $v_\alg(n,N,x) \le F_n(x) \cdot N$.	
\end{lemma}
\begin{proof}
We prove the statement by induction on $n$.

The base case is when $n=1$. Let $V_0$ be the set of initial vertices and $B$ be the set of the newly arrived vertices. According to our construction, $|B|=|V_0|=N$ and hence $|I| = 2N$ after the arrival of $B$. Observe that the algorithm can match at most $(1 - x) \cdot N$ fractional edges between $V_0$ and $B$, so that the average matched portion of the initial vertices $x_{V_0}$ is at most $1$. Thus,
\[
v_\alg(1,N,x) \le \nicefrac{1}{2} \cdot x\cdot N + (1-x)\cdot N -\Gamma \cdot N = (1-\Gamma - \nicefrac{x}{2}) \cdot N = F_1(x) \cdot N~,
\]
where the $\nicefrac{1}{2}\cdot x \cdot N$ term comes from the initial matched portion of $V_0$.

Next, we assume the statement holds for $n\ge 1$, and consider the case with $n+1$ steps. 
Let $V_0$ be a set of $N \ge \eps^{-(n+1)}$ isolated vertices whose average matched portion being $x$.
According to our construction, all these vertices belong to the same partition of active vertices $A$, and a corresponding set of vertices $B$ arrive online, where $A$ and $B$ forms a complete bipartite graph and $|B|=\eps|A|$.
Suppose the algorithm matches $a \cdot |B|$ fractional edges. 
We shall calculate the size of the matching produced by the algorithm vertex-by-vertex. Moreover, we do it in a lazy manner in the sense that we only count the contribution of a vertex at the moment when it becomes inactive. We also update the size of the optimal matching only when vertices become inactive.

Our construction continues depending on the two cases.

\paragraph{Against Aggressive Algorithms.}
According to our construction, we add $B$ as an active partition and recurse on $A$ and $B$ respectively. 
Let $\alg_A$ and $\alg_B$ be the two algorithms that start with the configuration $(n,N,x_A=x+\eps a)$ and $(n,\eps N, x_B=a)$, and follow the matching decision as $\alg$. Then we have,
\begin{equation}
v_\alg(n+1,N,x)= v_{\alg_A}(n,N,x+\eps a) + v_{\alg_B}(n,\eps N, a) \le F_n(x+\eps a)\cdot N + F_n(a)\cdot \eps N~,
\label{eqn:aggressive}
\end{equation}
where the inequality holds by induction hypothesis.

\paragraph{Against Conservative Algorithms.} Let $C$ be the $|B|$ least matched vertices of $A$. According to our construction, we deactivate $C\cup B$ and recurse on the active set $A\setminus C$. Therefore,
\[
v_\alg(n+1,N,x) = v_{\alg'}(n,(1-\eps)N,x_{A\setminus C}) + \left( \frac{x_C+a}{2} - \Gamma \right) \cdot \eps N~,
\]
where $\alg'$ denotes the algorithm that starts with the configuration $(n,(1-\eps)N,x_{A\setminus C})$, and follows the same matching decision as $\alg$. 

Applying induction hypothesis to $\alg'$, we have that
\[
v_{\alg'}(n,(1-\eps)N,x_{A\setminus C}) \le F_n(x_{A\setminus C}) \cdot (1-\eps)N~.
\]
Consequently, we have that
\begin{align}
& v_{\alg}(n+1,N,x) \le F_n(x_{A\setminus C}) \cdot (1-\eps)N + \left( \frac{x_C + a}{2} - \Gamma \right) \cdot \eps N \notag \\
& \le \left( F_n(x+\eps a) + \frac{1}{2} \cdot (x_{A\setminus C} - x - \eps a) \right) \cdot (1-\eps) N + \left( \frac{x_C + a}{2} - \Gamma \right) \cdot \eps N \notag \\
& = F_n(x+\eps a) \cdot (1-\eps)N + \frac{1}{2} \cdot \left( x_{A\setminus C} \cdot (1-\eps)N + x_C \cdot \eps N - (x+\eps a) \cdot (1-\eps) N \right) + \left(\frac{a}{2} - \Gamma \right) \cdot \eps N \notag \\
& = F_n(x+\eps a) \cdot (1-\eps)N + \frac{1}{2} \cdot \left( (x+\eps a)N - (x+\eps a) \cdot (1-\eps) N \right) + \left(\frac{a}{2} - \Gamma \right) \cdot \eps N \notag \\
& = \left( F_n(x+\eps a) \cdot (1-\eps) + \left( \frac{x+(1+\eps)a}{2} - \Gamma \right) \cdot \eps \right) \cdot N~.
\label{eqn:conservative}
\end{align}
Here, the second inequality holds by the Lipschitzness of $F_n$ (Claim~\ref{cl:lipschitz}) and the fact that $x_{A\setminus C} \ge x+\eps a$, since $C$ contains the least matched vertices of $A$; the second equality holds by the fact that $x_{A\setminus C} \cdot (1-\eps)N + x_C \cdot \eps N$ is the total matched portion of $A$, which equals $(x+\eps a)\cdot N$.

Recall that we diverge into the two cases by comparing \eqref{eqn:aggressive} and \eqref{eqn:conservative}, and apply the construction with the smaller value. This concludes the proof:  
\begin{align*}
v_\alg(n+1,N,x) & \le \min \bigg\{F_n(x+\eps a) \cdot (1-\eps) + \left(\frac{x+(1+\eps)a}{2} - \Gamma \right) \cdot \eps, \\
& \phantom{\le \min \bigg\{} \  F_n(x+\eps a) + \eps F_n(a) \bigg\} \cdot N\le F_{n+1}(x) \cdot N~.		
\end{align*}
\end{proof}

\section{Proof of Theorem~\ref{thm:main}}
\label{sec:main}
Finally, we put things together to complete the proof of our main theorem. Fix an $\eps>0$ so that $\eps^{-1}$ is an integer, and $\Gamma > \Gamma^*+\eps$. 
According to Lemma~\ref{lem:math}, there exists $n$ such that $F_n(0) <0$. 

Consider an instance with $N = \eps^{-n}$ isolated vertices arrived in a sequence at the beginning. There is no decision for the algorithm to make. 
By Lemma~\ref{lem:construction}, for any algorithm $\alg$, we have that
\[
v_\alg(n,N,0) \le F_n(0) \cdot N < 0~.
\]
That is, the fractional matching produced by $\alg$ is strictly smaller than $\Gamma$ times the size of the optimum matching, by the definition of $v_\alg$. In other words, no algorithm has a competitive ratio strictly greater than $\Gamma^* + \eps$. As this holds for every $\eps>0$, no algorithm can achieve a competitive ratio strictly greater than $\Gamma^*$. This concludes the proof of the theorem. 
\newpage

\bibliographystyle{plain}
\bibliography{matching}

\end{document}